\documentclass{fundam}%

\usepackage{amsfonts,amssymb}%
\usepackage{amsmath,amsxtra,amsfonts,latexsym,amscd,amssymb, amstext}
\usepackage{xspace}

\usepackage{graphicx}

\newcommand{\EoP}{\hbox{}\hfill\squareforqed\hbox{}}%

\newcommand{\defin}[1]{Definition~\ref{d.#1}}

\newcommand{\equnm}[1]{(\ref{q.#1})\xspace}

\newcommand{\figur}[1]{Figure~\ref{f.#1}}

\newcommand{\propo}[1]{Proposition~\ref{p.#1}}
\newcommand{\remar}[1]{Remark~\ref{r.#1}}
\newcommand{\secti}[1]{Section~\ref{s.#1}}
\newcommand{\theor}[1]{Theorem~\ref{t.#1}}

\newcounter{theor-tmp}
  {\setcounter{theor-tmp}{\value{theorem}}%
   \setcounter{theorem}{\value{#1}}%
   \begin{example}[continued]}%
  {\end{example}%
   \setcounter{theorem}{\value{theor-tmp}}}
\newcounter{exa-Ao}%
\newcommand{\LatinLocution}[1]{{\itshape #1}\xspace}
\newcommand{\cf}{\LatinLocution{cf.}}

\newcommand{\eg}{\LatinLocution{e.g.}}

\newcommand{\etc}{\LatinLocution{etc.}} 

\newcommand{\ie}{{that is, }}

\newcommand{\DSA}{DSA\xspace}
\newcommand{\PCSA}{PCSA\xspace}
\newcommand{\FPCSA}{FPCSA\xspace}
\newcommand{\e}{\text{\quad}}                 
\newcommand{\ee}{\text{\qquad}}               
\newcommand{\eee}{\text{\qquad \qquad}} 
\newcommand{\xmd}{\hspace{0.125em}} 
\newcommand{\msp}{\hspace{0.2em}} 
\newcommand{\quantvrg}{\, , \;}
\newcommand{\quantsp}{\ee }
\newcommand{\quantsmsp}{\e }
\newcommand{\eqpnt}{\makebox[0pt][l]{\:.}}%
\newcommand{\eqvrg}{\makebox[0pt][l]{\:,}}%
\newcommand{\EqVrgInt}{\: , \e }
\newcommand{\Defi}[2]%
    {\left\{#1\xmd\xmd\middle|\xmd\xmd#2\right\}}
\newcommand{\B}{\mathbb{B}}
\newcommand{\K}{\mathbb{K}}
\newcommand{\N}{\mathbb{N}}

\newcommand{\Z}{\mathbb{Z}}
\newcommand{\zeK}{0_{\K}}
\newcommand{\unK}{1_{\K}}
\newcommand{\aut}[1]{\left\langle\thinspace#1\thinspace\right\rangle}%
\newcommand{\auta}[1]{\left(\thinspace#1\thinspace\right)}%
\newcommand{\autiet}{\aut{I,E,T}}
\newlength{\vbh}\newlength{\vbd}\newlength{\vbt}%
\newcommand{\CompAut}[2]%
    {%
     \settodepth{\vbd}{\mbox{$\displaystyle{#1#2}$}}%
     \settoheight{\vbh}{\mbox{$\displaystyle{#1#2}$}}%
     \setlength{\vbt}{\vbh}\addtolength{\vbt}{\vbd}%
     {}%
     \psline[linewidth=0.8pt]{c-c}(0,-.65\vbd)(0,.9\vbh)%
     \hspace*{.15em}%
     {#1}%
     \hspace*{.15em}%
     \psline[linewidth=0.8pt]{c-c}(0,-.65\vbd)(0,.9\vbh)%
     }%

\newcommand{\ShiftInd}[1]{\raisebox{-0.3ex}{$\scriptstyle{#1}$}}
\newcommand{\Fut}[2]{{\mathsf{Fut}}_{\ShiftInd{#1}}(#2)}
\newcommand{\futa}[1]{\Fut{\Ac}{#1}}
\newcommand{\Ac}{\mathcal{A}}
\newcommand{\Bc}{\mathcal{B}}

\newcommand{\Aco}{\Ac_{1}}

\newcommand{\Am}{A_\mmark}

\newcommand{\mafig}[1]%
   {\refstepcounter{figure}\\
    \centering{\small\textbf{\figurename\thinspace\thefigure.} #1}}
\newcommand{\matab}[1]%
   {\refstepcounter{table}\\
    \centering{\small\textbf{\tablename\thinspace\thetable.} #1}}
\newcommand{\x}{\! \times \!}

\newcommand{\ab}{\{a,b\}}
\newcommand{\abe}{\ab^{*}}
\newcommand{\Ae}{A^{*}}
\newcommand{\Be}{B^{*}}
\newcommand{\Cc}{\mathbin{\mathcal{C}}}
\newcommand{\Pc}{\mathbin{\mathcal{P}}}
\newcommand{\Rc}{\mathbin{\mathcal{R}}}

\newcommand{\grando}[1]{\mathrm{O}\!\left(#1\right)}
\newcommand{\gtheta}[1]{\mathrm{\Theta}\!\left(#1\right)}
\renewcommand{\phi}{\varphi}

\def\dol{{\mathchoice
{\hbox{{\small $\textstyle \$ $}}}%
{\hbox{{\small $\textstyle \$ $}}}%
{\hbox{$\scriptstyle \$ $}}%
{\hbox{$\scriptscriptstyle \$ $}}%
}}
\newcommand{\mmark}{\dol}
\newcommand{\Adol}{A_{\dol}}

\newcommand{\Augm}[1]{\widetilde{#1}}
\newcommand{\AgmA}{\Augm{\Ac}}
\newcommand{\ConjAuto}[1]{\overset{#1}{\Longrightarrow}}

\newcommand{\amlg}[1]{X_{#1}}%
\newcommand{\slct}[1]{Y_{#1}}%
\newcommand{\amphi}{\amlg{\varphi}}%
\newcommand{\ampsi}{\amlg{\psi}}%
\newcommand{\amphip}{\amlg{\varphi'}}%
\newcommand{\ampsip}{\amlg{\psi'}}%
\newcommand{\amome}{\amlg{\omega}}%
\newcommand{\amPc}{\amlg{\Pc}}%
\newcommand{\slPc}{\slct{\Pc}}%

\newcommand{\Trsp}[1]{#1^{\mathsf{t}}}
\newcommand{\matmul}{\!\mathbin{\cdot}\!}
\newcommand{\mmul}{\matmul}
\newcommand{\ccdot}%
   {\mathbin{\raisebox{0.2ex}{${\scriptscriptstyle\circ}$}}}
\newcommand{\compos}{\ccdot}

\newcommand{\code}[1]{\textsl{#1}}
\newcommand{\keyw}[1]{\textsf{#1}}

\DeclareMathOperator{\meetop}{\code{meet}}
\newcommand{\meet}[1]{\meetop[#1]}
\DeclareMathOperator{\sigop}{\code{sig}}
\DeclareMathOperator{\glsigop}{\code{GSig}}

\newcommand{\sigf}[1]{\sigop\,[#1]}
\newcommand{\sig}[2]{\sigop\,[#1](#2)}

\newcommand{\glsigg}[1]{\glsigop[#1]}
\newcommand{\GSig}[1]{\glsigg{#1}} 

\DeclareMathOperator{\spltop}{\code{split}}
\newcommand{\spltf}[1]{\spltop[#1]}
\newcommand{\splt}[2]{\spltop[#1](#2)}
\newcommand{\splteq}{\mathbin{\wedge}}
\newcommand{\matriceuu}[1]%
    {\begin{pmatrix} #1 \end{pmatrix}}
\newcommand{\matricedd}[4]%
    {\begin{pmatrix} #1 & #2 \\ #3 & #4 \end{pmatrix}}
\newcommand{\vecteurd}[2]%
    {\begin{pmatrix} #1 \\ #2 \end{pmatrix}}
\newcommand{\ligned}[2]%
    {\begin{pmatrix} #1 & #2 \end{pmatrix}}
\newcommand{\matricett}[9]%
    {\begin{pmatrix}  #1 & #2 & #3 \\
                      #4 & #5 & #6 \\
                      #7 & #8 & #9 \end{pmatrix}}
\newcommand{\vecteurt}[3]%
    {\begin{pmatrix} #1 \\ #2 \\ #3 \end{pmatrix}}
\newcommand{\lignet}[3]%
    {\begin{pmatrix} #1 & #2 & #3 \end{pmatrix}}

\newlength{\jsWidthCol}
\newlength{\blocinterligne}
\newlength{\blocinterligned}
\newlength{\temparraycolsep}
\newlength{\longueurbloc}
\newlength{\hauteurbloc}
\newlength{\centragebloc}
\newlength{\longueurblc}
\newlength{\hauteurblc}
\newlength{\centrageblc}
\newcommand{\blocligne}[1]%
    {\framebox[\longueurbloc]{$#1$}}
\newcommand{\blocmatrice}[1]%
    {\framebox[\longueurbloc]{\rule[\centragebloc]{0mm}{\hauteurbloc}$#1$}}
\newcommand{\blocvecteur}[1]%
    {\framebox{\rule[\centragebloc]{0mm}{\hauteurbloc}$#1$}}
\newcommand{\blcligne}[1]%
    {\framebox[\longueurblc]{$#1$}}
\newcommand{\blcmatrice}[1]%
    {\framebox[\longueurblc]{\rule[\centrageblc]{0mm}{\hauteurblc}$#1$}}
\newcommand{\blcvecteur}[1]%
    {\framebox{\rule[\centrageblc]{0mm}{\hauteurblc}$#1$}}
\newcommand{\matriceddblvs}[4]
   {\setlength{\temparraycolsep}{\arraycolsep}%
    \setlength{\arraycolsep}{1.3pt}%
        \renewcommand{\arraystretch}{1.2}%
        \left (%
    \begin{array}{cc}%
                #1  & \blcligne{#2} \\
            \blcvecteur{#3} & \blcmatrice{#4}
        \end{array}%
        \right )%
    \setlength{\arraycolsep}{\temparraycolsep}%
        \renewcommand{\arraystretch}{1.0}%
   }%
\newcommand{\vecteurdblvs}[2]%
   {\setlength{\temparraycolsep}{\arraycolsep}%
    \setlength{\arraycolsep}{1.5pt}%
        \renewcommand{\arraystretch}{1.2}%
        \left (%
    \begin{array}{c}%
                #1  \\
            \blcvecteur{#2}
        \end{array}%
        \right )%
    \setlength{\arraycolsep}{\temparraycolsep}%
        \renewcommand{\arraystretch}{1.0}%
   }%
\newcommand{\lignedblvs}[2]%
   {\setlength{\temparraycolsep}{\arraycolsep}%
    \setlength{\arraycolsep}{1.5pt}%
        \renewcommand{\arraystretch}{1.2}%
        \left (%
    \begin{array}{cc}%
                #1  & \blcligne{#2}
        \end{array}%
        \right )%
    \setlength{\arraycolsep}{\temparraycolsep}%
        \renewcommand{\arraystretch}{1.0}%
   }%
\newcommand{\matricettblvs}[9]
   {\setlength{\temparraycolsep}{\arraycolsep}%
    \setlength{\arraycolsep}{1.5pt}%
        \renewcommand{\arraystretch}{1.2}%
        \left (%
    \begin{array}{ccc}%
                #1  & \blcligne{#2} & #3\\
            \blcvecteur{#4} & \blcmatrice{#5} & \blcvecteur{#6}\\
                #7  & \blcligne{#8} & #9\\
        \end{array}%
        \right )%
    \setlength{\arraycolsep}{\temparraycolsep}%
        \renewcommand{\arraystretch}{1.0}%
   }%
\newcommand{\vecteurtblvs}[3]%
   {\setlength{\temparraycolsep}{\arraycolsep}%
    \setlength{\arraycolsep}{1.5pt}%
        \renewcommand{\arraystretch}{1.2}%
        \left (%
    \begin{array}{c}%
                #1  \\
            \blcvecteur{#2}\\
                #3
        \end{array}%
        \right )%
    \setlength{\arraycolsep}{\temparraycolsep}%
        \renewcommand{\arraystretch}{1.0}%
   }%
\newcommand{\lignetblvs}[3]%
   {\setlength{\temparraycolsep}{\arraycolsep}%
    \setlength{\arraycolsep}{1.5pt}%
        \renewcommand{\arraystretch}{1.2}%
        \left (%
    \begin{array}{ccc}%
                #1  & \blcligne{#2} & #3
        \end{array}%
        \right )%
    \setlength{\arraycolsep}{\temparraycolsep}%
        \renewcommand{\arraystretch}{1.0}%
   }%
\newcommand{\matricettblblvs}[9]
   {\setlength{\temparraycolsep}{\arraycolsep}%
    \setlength{\arraycolsep}{1.5pt}%
        \renewcommand{\arraystretch}{1.2}%
        \left (%
    \begin{array}{ccc}%
                #1  & \blcligne{#2} & \blcligne{#3}\\
            \blcvecteur{#4} & \blcmatrice{#5} & \blcmatrice{#6}\\
                \blcvecteur{#7}  & \blcmatrice{#8} & \blcmatrice{#9}\\
        \end{array}%
        \right )%
    \setlength{\arraycolsep}{\temparraycolsep}%
        \renewcommand{\arraystretch}{1.0}%
   }%
\newcommand{\vecteurtblblvs}[3]%
   {\setlength{\temparraycolsep}{\arraycolsep}%
    \setlength{\arraycolsep}{1.5pt}%
        \renewcommand{\arraystretch}{1.2}%
        \left (%
    \begin{array}{c}%
                #1  \\
            \blcvecteur{#2}\\
                \blcvecteur{#3}
        \end{array}%
        \right )%
    \setlength{\arraycolsep}{\temparraycolsep}%
        \renewcommand{\arraystretch}{1.0}%
   }%
\newcommand{\lignetblblvs}[3]%
   {\setlength{\temparraycolsep}{\arraycolsep}%
    \setlength{\arraycolsep}{1.5pt}%
        \renewcommand{\arraystretch}{1.2}%
        \left (%
    \begin{array}{ccc}%
                #1  & \blcligne{#2} & \blcligne{#3}
        \end{array}%
        \right )%
    \setlength{\arraycolsep}{\temparraycolsep}%
        \renewcommand{\arraystretch}{1.0}%
   }%
\newcommand{\matblco}[3]{\matricettblvs{0}{#1}{0}%
                                   {0}{#2}{#3}%
								   {0}{0}{0}}
\newcommand{\matblct}[1]{\matricettblvs{1}{0}{0}%
                                   {0}{#1}{0}%
								   {0}{0}{1}}

\begin{document}

\setcounter{page}{195}
\publyear{22}
\papernumber{2126}
\volume{186}
\issue{1-4}

       \finalVersionForARXIV

\title{Morphisms and Minimisation of  Weighted  Automata}

\author{Sylvain Lombardy\\
LaBRI - UMR 5800 - Bordeaux INP - Bordeaux University - CNRS\\
Bordeaux,  France\\
Sylvain.Lombardy@labri.fr
\and
Jacques Sakarovitch\thanks{Address for correspondence: IRIF - CNRS/Paris Cit\'{e} University  and  Telecom Paris, IPP.}
 \\
IRIF - CNRS/Paris Cit\'{e} University  and  Telecom Paris, IPP\\
Paris, France\\
sakarovitch@enst.fr
}

\runninghead{S. Lombardy and J. Sakarovitch}{Morphisms and Minimisation of  Weighted  Automata}

\maketitle

\begin{abstract}
This paper studies the algorithms for the minimisation of weighted
automata. It starts with the definition of morphisms --- which generalises
and unifies the notion of  \textit{bisimulation} to the whole class of
weighted automata --- and the unicity of a minimal quotient for every
automaton, obtained by partition refinement.

From a general scheme for the refinement of partitions, two
strategies are considered for the computation of the minimal
quotient: the \textit{Domain Split} and the \textit{Predecesor Class Split} algorithms.
They correspond respectivly to the classical Moore and Hopcroft
algorithms for the computation of the minimal quotient of
deterministic Boolean automata.%

We show that these two strategies yield algorithms with the same
quadratic complexity and we study the cases when the second one can be
improved in order to achieve a complexity similar to the one of
Hopcroft algorithm.
\end{abstract}

\section{Introduction}%
\label{s.int}%

The main goal of this paper is to report on the design and
analysis of algorithms for the computation of the minimal quotient of
a weighted automaton.

The \emph{existence} of a minimal deterministic finite automaton,
canonically associated with every regular language is one of the
basic and fundamental results in the theory of classical finite
automata~\cite{HopcMotwUllm06}.
The problem of the \emph{computation} of this minimal (deterministic)
automaton has given rise to an extensive literature, due to the
importance of the problem, both from a theoretical and practical
point of view.
The chapter \emph{Minimisation of automata} of the recently published
\textit{Handbook of Automata Theory}~\cite{Pin21Ed}
is devoted to this question by
Jean~Berstel and his colleagues~\cite{BersEtAl21}.
It provides a rich and detailed account on the subject, together with
an extensive list of references.

In contrast, the problem of minimisation of nondeterministic Boolean,
or of weighted, automata is much less documented.
The chapter \emph{Algorithms for weighted automata} in the
\textit{Handbook of Weighted Automata}~\cite{DrosEtAl09Ed} deals only
with the case of so-called deterministic weighted automata on which
the algorithms for deterministic Boolean automata can be generalised.
The main reason is probably that the problem of finding a (Boolean)
automaton equivalent to a given automaton and with a minimal number
of states is untractable and known to be NP-hard and even PSPACE-complete~\cite{GareJohn79}.
 Of course, the case of weighted automata, with arbitrary weight
semiring, is at least as difficult.

There is another way to look at the minimisation of deterministic
Boolean automata.
The result, the minimal automaton, is obtained by merging states and
in this way can be seen as the image of the original automaton by a
map that preserves the structure of the computations, a map it is
natural to call {morphism}.
The other purpose of this paper, which indeed comes first, is then to
set up the definition of \emph{morphisms for arbitrary (finite) weighted
automata}, and in particular for nondeterministic Boolean automata.
The image of a morphism is a \emph{quotient} and, as in the original
case, an automaton has a \emph{unique minimal quotient}
(Theorem 3.12). 
The difference with the original case is that the minimal quotient is
not canonically attached to the series or the language realised by
the automaton but to the automaton itself.

To tell the truth, this point of view is not completely new.
For transition systems for instance, which are
automata without initial and final states,
it is common knowledge that the notion of
\emph{bisimulation} allows to merge states in order to get a system
in which the computations are preserved and the \emph{coarsest
bisimulation} yields a minimal system.
The notion of bisimulation has also been extended to some classes of
weighted automata, for instance
\emph{probabilistic
automata}~\cite{BaieHerm97}, or automata with weights in a field or
division ring~\cite{Bore09}.
It is also known that in all these cases the coarsest bisimulation is
computed by iterative refinements of set partitions.

We show here that the classical algorithms of partition refinement
that are used for deterministic Boolean automata (widely known as
Moore and Hopcroft algorithms) may be analysed and abstracted in such
a way they readily extend to weighted automata in full generality,
without any assumption
on the weight semiring and on the automaton structure.
This can be sketched as follows. In a partition refinement algorithm, and at a given step of the
procedure, a partition~$\Pc$ of the state set of an automaton, and
a class~$C$ of~$\Pc$ are considered. There are then two possible strategies for determining a refinement
of~$\Pc$. In the first one, the class~$C$ itself is split, by considering the
labels of the \emph{outgoing transitions} from the different states
in~$C$. This is an extension of the Moore algorithm and we call it the
\emph{Domain Split Algorithm}. In the second strategy, the class~$C$ determines the splitting of
classes that contain the origins of the \textit{transitions incoming}
to the states in~$C$.
We call it the \emph{Predecessor Class Split Algorithm} and it can be
seen as inspired by the Hopcroft algorithm.

\eject

Although the two strategies yield distinct orderings of the splitting
of classes, the two algorithms have many similarities that we describe
in this paper.
Not only do they have the same time complexity, in~$\grando{n\xmd m}$,
where~$n$ is the number of states of the automaton and~$m$ the number
of its transitions, but the criterium for distinguishing states ---
the splitting process --- is based on the same state function that we
call \emph{signature}.
And in both cases, achieving the above mentioned complexity implies
that signatures are managed through the same efficient data structure
that implements a \emph{weak sort}.

\enlargethispage*{3ex}
Finally, our analysis allows to describe a condition --- which we call
\emph{simplifiable signatures} --- under which the Predecessor Class
Split Algorithm can be tuned in such a way it achieves a better time
in~$\grando{m\xmd \log n}$, the Graal set up by Hopcroft algorithm for
deterministic Boolean automata (in which~$m=\alpha\xmd n$,
where~$\alpha$ is the size of the alphabet).

In conclusion, our study of the minimisation algorithms, with the
identification of the concept of signature, allows to reverse the
perspective.
We have not transformed the known algorithms on deterministic Boolean
automata by equipping them with supplementary features that allow to
treat weighted automata, we have just expressed them in a way they can
deal with all weighted automata and they apply then to deterministic
Boolean automata just as they do with any other, not even in a simpler
way.

\medskip
The paper is organised as follows.
In \secti{wei-aut}, we begin with the definitions and notation used
for weighted automata.
We add the definition of a special class of automata, a kind of
normalised ones, which we call \emph{augmented automata} and which
will be useful for the writing of algorithms.
At \secti{min-quo-aut}, we give two equivalent definitions for
morphisms of weighted automata and show the existence of a unique
minimal quotient.

In Section 4, we describe a general procedure for the
partition refinement, the notion of signature, and both the
Domain Split and Predecessor Class Split Algorithms.
In the next \secti{imp-alg}, we describe the way to implement the weak
sort, how it is used in the two algorithms, and we analyse their
complexity,
whereas we explain in \secti{fas-PCS} under which
conditions and how it is possible to improve the complexity of the
latter.

The algorithms and their efficient implementation described in this
paper are accessible in the \textsc{Awali} plateform recently made
available to the public~\cite{Awali}.
Some experiments presented at \secti{ben-mar} show that the
algorithms behave with their theoretical complexity.
These experiments were presented at the CIAA conference in 2018~\cite{LombSaka18}.

\section{Weighted automata} 
\label{s.wei-aut}%

In this paper, we deal with (finite) automata over a free monoid~$\Ae$
with weight in a semiring~$\K$, also called $\K$-automata.
`Classical' automata are the $\B$-automata where~$\B$ is the Boolean
semiring.

Indeed, all what follows apply as well to automata over a monoid~$M$
which is not necessarily free, for instance to \emph{transducers} that
are automata over the non free monoid~$\Ae\x\Be$.
This generalisation holds as such automata are considered --- as far
as the constructions and results developped here are concerned --- as
automata over a free monoid~$C^{*}$, where~$C$ is the set of labels on
the transitions of the automaton.
Let us note that there exists no theory of quotient that takes into
account non trivial relations between labels.
The only construction on automata that does (without destroying their
structure) is the \emph{circulation of labels} involved for instance
in the \emph{synchronisation} of transducers
(\cf~\cite{Saka09,FrouSaka93}) or as a \emph{preliminary} for the
minimisation of sequential transducers (\cf~\cite{Saka09,Reut90}).

\subsection{Definition and notation}

We essentially follow the definitions and notation of~\cite{Saka09}.
The model of weighted automaton used in this paper is more restricted
though, for both theoretical and computational efficiency.

\medskip
A $\K$-automaton~$\Ac$ over~$\Ae$ is a directed graph whose vertices
are called \emph{states} and whose edges, called \emph{transitions},
are labelled by a pair made of a letter in~$A$ and a weight in~$\K$,
together with an \emph{initial function} and a \emph{final function},
both from the set of states to~$\K$.
The states in the support of the initial function are usually called
\emph{initial states}, those in the support of the final function,
\emph{final states.}
\figur{ex1} shows a $\Z$-automaton~$\Ac_{1}$ over~$\abe$ with
the usual graphical conventions.

 \begin{figure}[!h]
\vspace*{1mm}
\begin{center}
\includegraphics[width=5.22cm]{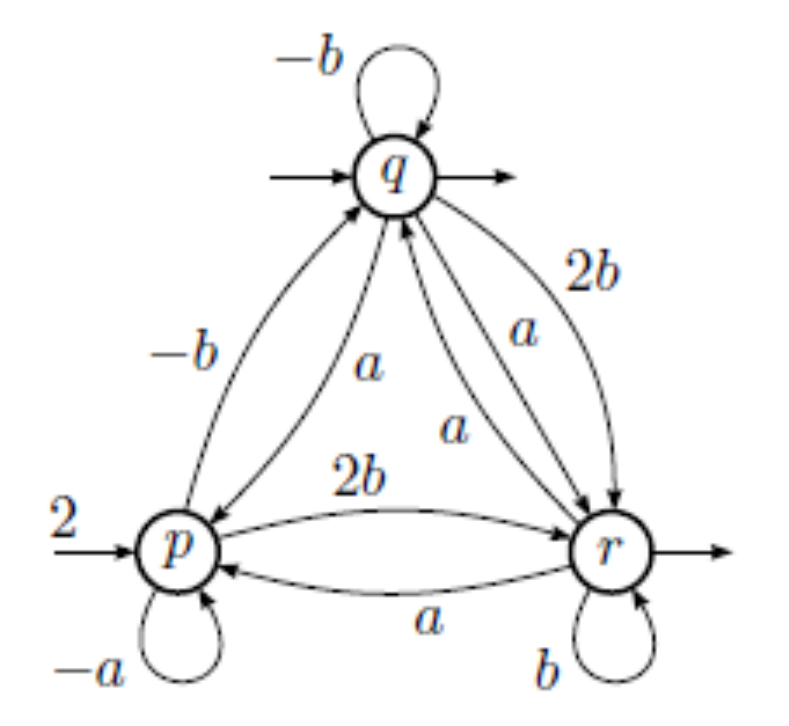}\vspace*{-3mm}
\caption{The $\Z$-automaton~$\Ac_{1}$}
\label{f.ex1}%
\end{center}\vspace*{-2mm}
\end{figure}

Let~$\Ac$ be a $\K$-automaton with set of states~$Q$; we say that~$Q$
is the \emph{dimension} of~$\Ac$ and we denote~$\Ac$ by a triple
$\msp\Ac=\aut{I,E,T}\msp$ where~$I$ and~$T$ are vectors of
dimension~$Q$ with entries in~$\K$ that denote the initial function
and the final functions and~$E$ is the adjacency matrix of~$\Ac$,
that is, a $Q\x Q$-matrix whose~$(p,q)$ entry is the sum of the
weighted labels of the transitions from~$p$ to~$q$.
We also denote by~$E$ the map $\msp E\colon Q\x A\x Q \to\K\msp$
and~$E(p,a,q)$ is the weight of the (unique) transition that goes
from~$p$ to~$q$ and that is labelled by~$a$, if it exists (otherwise
its value is~$\zeK$).

\begin{example}
Let~$\Aco=\aut{I_{1}, E_{1}, T_{1}}$ of dimension~$Q_{1}=\{p,q,r\}$:
\begin{equation}
\Ac_{1} = \aut{\lignet{2}{1}{0},%
               \matricett{-a}{-b}{2\xmd b}%
				         {a}{-b}{a + 2\xmd b}%
						 {a}{a}{b},%
			   \vecteurt{0}{1}{1}}
\qquad  \mbox{and}\quad E_1 (q, b, r) =2 \eqpnt
\notag
\end{equation}
\end{example}

\medskip

The \emph{behaviour} of a $\K$-automaton~$\Ac=\autiet$ is the
\emph{series} realised by~$\Ac$ and is denoted by~$|\Ac|$:
\begin{equation}
|\Ac|= I \cdot E^{*} \cdot T
\notag
\end{equation}
where $\msp E^{*} = \sum_{n\in\N} E^{n}\msp$ (the infinite sum does
not bring any problem as our definition insures that the entries
of~$E^{n}$ are homogenous polynomials of degree~$n$).
Two automata are said to be \emph{equivalent} if their behaviours are
equal. If $\K=\B$, $|\Ac|$ is the \emph{language} accepted by~$\Ac$.

The \emph{future} of a state~$p$ of~$\Ac$, denoted by~$\futa{p}$, is
the series realised by~$\Ac$ when~$p$ is taken as the unique initial
state, with initial value equal to~$\unK$. It holds
\begin{equation}
\futa{p} = I_{p}\matmul E^{*}\matmul T
\eqvrg
\label{q.fut-1}
\end{equation}
where~$I_{p}$ denotes the characteristic (row-) vector of~$p$, that
is, the entry of index~$p$ is the only non-zero entry and is equal
to~$\unK$.

\subsection{The augmented automaton}

The algorithms we describe and study in the following sections
classify states of an automaton according to (the labels of) the in-
and out-going transitions and (the values of) the initial and final
functions.
In order to express in the same way the conditions stated on the
transitions on one hand and on the initial and final functions on the
other,
and to compute uniformly with the former and the latter, it is
convenient to transform the automaton and to put them in a special form
which we call \emph{augmented}.

\begin{figure}[!b]\small
\arraycolsep=1em
\centering
$\begin{array}{c|cccc}
 & p & q & r & t \\
\hline
i & 2\xmd\mmark & \mmark & 0 & 0\\
p & -a & -b & 2\xmd b & 0 \\
q & a & -b & a+2\xmd b & \mmark\\
r & a & a & b  & \mmark\\
\end{array}$\vspace*{-2mm}
\caption{$\Augm{E}_{1}$, the adjacency matrix of~$\AgmA_{1}$}
\label{f.ex1b}\vspace*{-1mm}
\end{figure}

\smallskip
A $\K$-automaton $\Ac=\aut{I,E,T}$ of dimension~$Q$ over~$\Ae$ is
transformed into its augmented version, denoted by~$\AgmA$, in two
steps.
First, the alphabet~$A$ is supplemented with a new letter~$\dol$,
which will be used as a left and right marker.
We write~$\Adol$ for $\msp\Adol=A\cup\{\dol\}\msp$, the
\emph{augmented alphabet}.
Second, $\Ac$ is somehow \emph{normalised} by the adjunction of two
new states~$\xmd i\xmd$ and~$\xmd t\xmd$ to~$Q$ and of transitions
that go from~$i$ to every initial state~$p$ with label~$\dol$ and with
weight~$I_{p}$ and transitions that go from every final state~$q$
to~$t$ with label~$\dol$ and with weight~$T_{q}$.
As it will be useful to make its state set explicit, we denote~$\AgmA$
by~$\AgmA=\auta{Q,i,\Augm{E},t}$ where~$\Augm{E}$ contains the
complete description of~$\AgmA$, as seen in~\equnm{aug}.
\begin{equation}
\Ac = \aut{I,E,T}
\ee\text{and}\ee
\AgmA = \auta{Q,i,\Augm{E},t}
\e\text{with}\e
\Augm{E} = \matblco{I\dol}{E}{T\dol}
\eqpnt
\label{q.aug}
\end{equation}
Finally, the only initial state of~$\AgmA$ is~$i$, with weight~$\unK$,
and its only final state is~$t$, also with weight~$\unK$.
If~$w$ is in~$\Ae$, there is a 1-1 correspondence between the
successful computations with label~$w$ in~$\Ac$ and the computations
with label~$\dol\xmd w\xmd\dol$ in~$\AgmA$, hence they are given the
same weight by the two automata.
Hence, $|{\AgmA}|=\$ |\Ac|\$$.

\begin{example}[continued]
The adjacency matrix~$\Augm{E}_{1}$ of the $\Z$-automaton~$\AgmA_{1}$ is shown
on \figur{ex1b}.
There is no column~$i$ nor row~$t$ in the table for there is
no transitions incoming to~$i$ nor transitions outgoing from~$t$.
\end{example}

\section{Morphisms of weighted automata and minimal quotients}
\label{s.min-quo-aut}%

The notion of morphism for deterministic Boolean automata does not
raise difficulties and makes consensus.
In contrast, the one of morphism for arbitrary Boolean automata and
even more for (arbitrary) weighted automata is far more problematic.
In most cases, the concept is given other names, most often
\emph{simulation} or \emph{bisimulation}, when not more complicated
such as \emph{weak bisimulation} or \emph{pure epimorphism}, and
definitions that may depend on the case considered and hide their
generality.

\subsection{Amalgamation matrices and morphisms}

We start with the definition of \emph{morphisms} --- of
\emph{Out-morphisms} indeed as we shall see --- by means of the notion
of \emph{conjugacy} of automata, as we already did in~\cite{Saka09}
or~\cite{BealEtAl06} and which is the most concise one.
We then translate it into a more explicit definition via equivalence
that is more suited for computations.

\begin{definition}
\label{d.con-jug}%
Let
$\msp\Ac=\aut{I,E,T}\msp$
and
$\msp\Bc=\aut{J,F,U}\msp$
be two $\K$-automata, of dimension~$Q$ and~$R$ respectively.
We say that \emph{$\Ac$ is conjugate to~$\Bc$ by~$X$} if there exists
a $Q\x R$-matrix~$X$ with entries in~$\K$ such that
\begin{equation}
I\matmul X=J\EqVrgInt
E\matmul X=X\matmul F\EqVrgInt\text{and} \e
T=X\matmul U
\eqpnt
\label{q.con-aut}
\end{equation}
The matrix~$X$ is the \emph{transfer matrix} of the conjugacy
and we write $\msp\Ac\ConjAuto{X}\Bc\msp$.
\end{definition}

If~$\Ac$ is conjugate to~$\Bc$,
then, for every~$n$, the sequence of equalities holds:
\begin{equation*}
I\mmul E^{n}\mmul T =
I\mmul E^{n}\mmul X\mmul U =
I\mmul E^{n-1}\mmul X\mmul F \mmul U = \ldots =
I\mmul X\mmul F^{n} \mmul U = J \mmul F^{n} \mmul U
\eqvrg
\end{equation*}
from which
$\msp I\mmul E^{*}\mmul T = J \mmul F^{*} \mmul U\msp$
directly follows.
And this is stated as the following.

\begin{proposition}
\label{p.cnj-equ}%
If $\msp\Ac\msp$ {is conjugate to} $\msp\Bc\msp$,
then $\msp\Ac\msp$ and  $\msp\Bc\msp$ are equivalent.
\EoP
\end{proposition}

The conjugacy relation is not an equivalence relation as it is
reflexive and transitive but not symmetric.

A \emph{surjective} map $\msp\varphi\colon Q\rightarrow R \msp$ is
completely described by the $Q\x R$-matrix~$\amphi$ whose
$(q,r)$-th entry is~$1$ if $\varphi(q)=r$, and~$0$ otherwise.
Since~$\varphi$ is a map, every row of~$\amphi$ contains exactly
one~$1$ and since~$\varphi$ is surjective, every column of~$\amphi$
contains at least one~$1$.
Such a matrix is called an \emph{amalgamation matrix} in the setting
of symbolic dynamics (\cite{LindMarc95}).
By convention, if we deal with $\K$-automata, an amalgamation matrix
is silently assumed to be a $\K$-matrix, that is, the null entries
are equal to~$\zeK$ and the non-zero entries to~$\unK$.

\begin{definition}
\label{d.out-mor}%
 Let $\msp\Ac=\aut{I,E,T}\msp$ and $\msp\Bc=\aut{J,F,U}\msp$ be two
 $\K$-automata of dimension~$Q$ and~$R$ respectively.
A surjective map $\msp\varphi\colon Q\rightarrow R \msp$
is a \emph{morphism} (from~$\Ac$ onto~$\Bc$) if~$\Ac$ is conjugate
to~$\Bc$ by~$\amphi$: $ \msp\Ac\ConjAuto{\amphi}\Bc\msp$, \ie
\begin{equation}
I\matmul \amphi=J\EqVrgInt
E\matmul \amphi=\amphi\matmul F\EqVrgInt\text{and} \e
T=\amphi\matmul U
\eqvrg
\label{q.mor-con}
\end{equation}
and
we write $\msp\varphi\colon\Ac\rightarrow\Bc\msp$.

\medskip
We also say that~$\Bc$ is a \emph{quotient} of~$\Ac$, if there exists
a morphism~$\msp\varphi\colon\Ac\rightarrow\Bc\msp$.
\end{definition}

It is straightforward that the composition of two morphisms is a
morphism.
And by \propo{cnj-equ}, any quotient of~$\Ac$ is equivalent to~$\Ac$.

\defin{out-mor} may be given a form which will probably be more
eloquent to many than a mere matrix equation, and which makes clear that
the automaton~$\Bc$, and its state set~$R$, are immaterial and that
the conditions are upon~$\Ac$ and the map equivalence of~$\varphi$
only.
Above all, it will be the one used in the design of the algorithms to
come.
We describe it in the next subsection.

\subsection{Morphisms and congruences}

We start with some definitions and notation to deal with
equivalences.
An \emph{equivalence} on a set~$Q$ is a \emph{partition} of~$Q$, that
is, a set of non-empty disjoint subsets
of~$Q$, called \emph{classes},
whose union is equal to~$Q$.
If~$\Pc$ is an equivalence on~$Q$, we denote by~$\Pc$ both the
\emph{set} of classes in the partition as well as the \emph{relation}
on~$Q$ determined by the partition, that is, for every pair~$(p,q)$ of
elements of~$Q$, $p\Pc q$ if and only~$p$ and~$q$ belong to a same
class~$C$ of~$\Pc$.

Equivalences and surjective maps are indeed the same objects, seen
from a slightly different perspective.
An equivalence~$\Pc$ on~$Q$ determines the surjective map that sends
every~$q$ in~$Q$ onto its class~$C$ modulo~$\Pc$ and conversely every
surjective map~$\phi$ determines the map equivalence.
In both cases, they are represented by amalgation matrices.

Let us introduce a last definition.
Let~$\Pc$ be an equivalence on~$Q$ and~$\amPc$ its amalgamation
matrix.
From~$\amPc$ we construct a \emph{selection matrix}~$\slPc$ by
transposing~$\amPc$ and by zeroing some of its non-zero entries in
such a way that~$\slPc$ is row-monomial, with \emph{exactly one}~$1$
per row.
A matrix~$\slPc$ is not uniquely determined by~$\Pc$ but also
depends on the choice of the entry which is kept equal to~$1$, that
is, of a `representative' in each class of~$\Pc$.
\begin{definition}
\label{d.con}
Let~$\Ac$ be an automaton of dimension~$Q$.
We call \emph{congruence} on~$Q$ the map equivalence of a
morphism~$\phi$ from~$\Ac$ onto an automaton~$\Bc$.
\end{definition}

\begin{proposition}
\label{p.con}
Let $\Ac=\autiet$ be a $\K$-automaton over~$\Ae$ of dimension~$Q$.
An equivalence~$\Pc$ on~$Q$ is a \emph{congruence} if and only if the
following holds:
\begin{align}
\forall p,q\in Q\quantsp
p\Pc q &\e\Longrightarrow\e
\forall a\in A\quantvrg 
\forall C\in\Pc\quantsmsp
\sum_{r\in C}E(p,a,r)=\sum_{r\in C}E(q,a,r)
\eqvrg
\e
\label{q.con-1}
\\
\forall p,q\in Q\quantsp
p\Pc q &\e\Longrightarrow\e
T_{p} = T_{q}
\eqpnt
\label{q.con-2}
\end{align}
\end{proposition}

\begin{proof}
\textsl{The condition is necessary.}
Let $\msp\Bc=\aut{J,F,U}\msp$ be a quotient of~$\Ac$,
$\msp\varphi\colon\Ac\rightarrow\Bc\msp$ a morphism, and~$\Pc$ its
map equivalence.
The multiplication of~$E$ \emph{on the right} by~$\amphi$ amounts to add
together the \emph{columns} of~$E$ whose indices are sent to a same
element by~$\phi$, that is, indices which are in a same class
of~$\Pc$.
The entries of the $Q\x\Pc$-matrix~$G=E\matmul\amphi$ are:
\begin{equation}
\forall p\in Q\quantvrg \forall C\in\Pc\quantsp
G_{p,C}=\sum_{r\in C}E_{p,r}
\eqpnt
\label{q.con-3}
\end{equation}
Every~$G_{p,C}$ is a linear combination of letters of~$A$:
$\msp G_{p,C}=\sum_{a\in A}  G(p,a,C)\xmd a\msp$
and~\equnm{con-3} may be rewritten as
\begin{equation}
\forall p\in Q\quantvrg \forall C\in\Pc\quantvrg \forall a\in A\quantsp
G(p,a,C) =\sum_{r\in C}E(p,a,r)
\eqpnt
\label{q.con-4}
\end{equation}

On the other hand, the multiplication of~$F$, \emph{on the left},
by~$\amphi$ yields a matrix whose \emph{rows} with indices in the
same class modulo~$\Pc$ are equal.
Hence the equality $\msp E\matmul\amphi=\amphi\matmul F\msp$
implies~\equnm{con-1}.
For the same reason, $\msp T=\amphi\matmul U\msp$
implies~\equnm{con-2}.

\medskip
\noindent
\textsl{The condition is sufficient.}
Let~$\Pc$ be an equivalence on~$Q$ such that~\equnm{con-1}
and~\equnm{con-2} hold, $\amPc$ its amalgamation matrix, and~$\slPc$
a selection matrix.

The entries of~$T$ are indexed by~$Q$, and~$\slPc\matmul T$ is a
(column-) vector of dimension~$\Pc$ obtained by picking one entry
of~$T$ in each class~$C$ modulo~$\Pc$.
The vector~$\amPc\matmul(\slPc\matmul T)$, of dimension~$Q$, is
obtained by replicating, for every~$C$ in~$\Pc$, the entry
of~$\slPc\matmul T$ indexed by~$C$ for each~$p$ of~$Q$ in~$C$.
Since~\equnm{con-2} expresses that all entries of~$T$ indexed by
states in a same class  are equal,
$\msp T=\amPc\matmul(\slPc\matmul T)\msp$ holds.

In the same way, the rows (of dimension~$\Pc$) of~$E\matmul\amPc$ are
indexed by~$Q$, and~$\slPc\matmul E\matmul\amPc$ is a
$\Pc\x\Pc$-matrix obtained by picking one row of~$E\matmul\amPc$ in
each class~$C$ modulo~$\Pc$.
The matrix~$\amPc\matmul(\slPc\matmul E\matmul\amPc)$, of
dimension~$Q\x\Pc$, is obtained by replicating, for every~$C$
in~$\Pc$, the row of~$\slPc\matmul E\matmul\amPc$ indexed by~$C$ for
each~$p$ of~$Q$ in~$C$.
Since~\equnm{con-1} expresses that all rows of~$E\matmul\amPc$
indexed by  states in a same class  are equal,
$\msp E\matmul\amPc=\amPc\matmul(\slPc\matmul E\matmul\amPc)\msp$
holds.

\medskip
Equations~\equnm{con-1} and~\equnm{con-2} express
then that~$\Ac$ is conjugate by~$\amPc$ to
the automaton
\begin{equation}
\aut{I\matmul\amPc, \msp\slPc\matmul E\matmul\amPc, \msp\slPc\matmul T}
\eqvrg
\notag
\end{equation}
hence the map $\msp\phi\colon Q\rightarrow\Pc\msp$ is a morphism
and~$\Pc$, its map equivalence, is a congruence.
\end{proof}

\begin{remark}
\label{r.mor-bis}%
As mentioned in the introduction, morphisms have a close relationship
with \emph{bisimulations}.
One can even say it is the same notion expressed differently, for
those weighted automata for which bisimulations have been defined:
bisimulations for \emph{transition systems} which are Boolean automata
without initial and final states~\cite{Miln89,BensBenSh88}, for
\emph{probabilistic systems}~\cite{LarsSkou91,BaieHerm97}, linear
bisimulations when the weight semiring is a field~\cite{Bore09}, \etc
and when it is freed from so-called \emph{internal actions}.

Indeed an equivalence relation on the state set of an automaton~$\Ac$
is a bisimulation relation if and only if it is a congruence.
And a state of~$\Ac$ is in bisimulation with its image in any
quotient of~$\Ac$.
\end{remark}

\begin{example}[continued]
Let~$\Aco=\aut{I_{1}, E_{1}, T_{1}}$ of dimension~$Q_{1}=\{p,q,r\}$:
\begin{equation}
\Ac_{1} = \aut{\lignet{2}{1}{0},%
               \matricett{-a}{-b}{2\xmd b}%
				         {a}{-b}{a + 2\xmd b}%
						 {a}{a}{b},%
			   \vecteurt{0}{1}{1}}
\eqpnt
\notag
\end{equation}
It is easily seen that if we add the columns~$q$ and~$r$ of~$E_{1}$
we get a matrix whose rows~$q$ and~$r$ are equal;
moreover~${T_{1}}_{q}={T_{1}}_{r}$.
Hence $\{\{p\},\{q,r\}\}$ is a congruence of~$Q_{1}$.
\end{example}

Since it is characteristic, we could have taken just as well the
statement of \propo{con} as the definition of a congruence or a
morphism.
And all the more in this paper where it is the property which is
uniformly called.
But in view of further references to this paper, we prefer to put
\defin{out-mor} in front, for its mathematical efficiency (\eg
\cite{LombSaka21}).

\subsection{Further properties of morphisms and congruences}

\begin{proposition}
\label{p.con-fut}
Let $\msp\Ac=\autiet\msp$ be a $\K$-automaton over~$\Ae$ of dimension~$Q$
and~$\xmd\Pc\xmd$ a {congruence} on~$Q$.
If two states~$p$ and~$q$ are equivalent modulo~$\Pc$, then
$\msp\futa{p}=\futa{q}\msp$.
\end{proposition}

\begin{proof}
With the notation of \secti{wei-aut},
$\msp\futa{p} = I_{p}\matmul E^{*}\matmul T\msp$ and with the same
computation as above,
$\msp\futa{p} = I_{p}\matmul\amPc\matmul F^{*}\matmul U\msp$.
Accordingly,
$\msp\futa{q} = I_{q}\matmul\amPc\matmul F^{*}\matmul U\msp$.
And if~$p$ and~$q$ are equivalent modulo~$\Pc$, then
$\msp I_{p}\matmul\amPc = I_{q}\matmul\amPc \msp$.
\end{proof}

\begin{remark}
\label{r.mor-dir}%
\defin{out-mor} gives a notion of morphism that is \emph{directed} as
the one of conjugacy is: $E$ is multiplied \emph{on the right} by~$X$
in \equnm{con-aut}.
This is even more obvious with the statements of \propo{con} or of
\propo{con-fut}.

Thus morphisms could, or should, be called \emph{Out-morphisms} as it
refers to \emph{out-going} transitions.
The same map~$\phi$ would be an \emph{In-morphism} if~$\Bc$ is
conjugate to~$\Ac$ by~$\Trsp{\amphi}$ (the transpose of~$\amphi$).

In this work, we deal with Out-morphisms only, which we simply call
morphisms.
All statements can be dualised and transformed accordingly for
In-morphisms.
\end{remark}

\begin{remark}
\label{r.mor-boo}%
\textsl{The case of Boolean automata}.\e
Let $\msp\Ac=\aut{I,E,T}\msp$ and $\msp\Bc=\aut{J,F,U}\msp$ be two
\emph{Boolean} automata over~$\Ae$ of dimension~$Q$ and~$R$
respectively and $\msp\varphi\colon Q\rightarrow R \msp$ a morphism
from~$\Ac$ to~$\Bc$.
In this special case, no weight is really involved: a transition
exists or not, a state is initial or not, final or not.
Moreover, $I$~and~$J$ can also be seen as subsets of~$Q$, $J$~and~$U$
as subsets of~$R$, $E$~can also be seen as a subset of~$Q\x A\x Q$,
$F$~as a subset of~$R\x A\x R$.
\defin{out-mor} can then be  rewritten in the following way:
$\msp I\matmul\amphi = J\msp$ and
$\msp T = \amphi\matmul U\msp$ translate into
\begin{equation}
\text{(i)} \ee \phi(I) = J
\eee\ee \text{and} \eee\ee
\text{(ii)} \ee  T = \phi^{-1}(U)
\EqVrgInt
\notag
\end{equation}
whereas $\msp E\matmul\amphi = \amphi\matmul F\msp$ yields
\begin{alignat}{2}
\text{(iii)} \e & \forall a\in A\quantvrg \forall p,q\in Q\quantsp
&(p,a,q)\in E \e&\Longrightarrow\e (\phi(p),a,\phi(q))\in F
\ee\text{and}\e
\notag\\
\text{(iv)} \e & \forall a\in A\quantvrg
  \forall p\in Q\quantvrg\forall s\in R\quantsp
&(\phi(p),a,s)\in F \e&\Longrightarrow\e
\exists q\in\phi^{-1}(s)\quantsmsp (p,a,q)\in E
\eqpnt
\notag
\end{alignat}

Compared with previous terminology, morphisms of Boolean automata as
we have just defined them are what we have called \emph{locally
surjective Out-morphisms} in~\cite{Saka09} and subsequent works.
In the case of \emph{deterministic} Boolean automata, our definition
coincide with the classical notion of morphisms of automata, when it
is used.
\end{remark}

\begin{remark}
\label{r.con-agm}%
\textsl{The case of augmented automata}.\e
If $\msp\varphi\colon\Ac\rightarrow\Bc\msp$ is a morphism, then
$\msp\Augm{\varphi}\colon\Augm{\Ac}\rightarrow\Augm{\Bc}\msp$ is also
a morphism, where $\AgmA=\auta{Q,i,\Augm{E},t}$,
$\Augm{\Bc}=\auta{R,j,\Augm{F},u}$, and
$\Augm{\varphi}(i)=j$, $\Augm{\varphi}(t)=u$, and
$\Augm{\varphi}(q)=\phi(q)$ for all~$q$ in~$Q$.
This amounts to say that~$\amlg{\Augm{\varphi}}$ has the form
\begin{equation}
\amlg{\Augm{\varphi}} = \matblct{\amphi}
\eqvrg
\label{q.aug-mor}
\end{equation}
a matrix which we rather denote by~$\Augm{\amphi}$.

\medskip
Conversely, if $\AgmA=\auta{Q,i,\Augm{E},t}$ is the augmented
automaton of $\Ac=\autiet$, the amalgamation matrix of \emph{any
morphism} of~$\AgmA$ is of the form~\equnm{aug-mor} and corresponds
to a morphism of~$\Ac$.
In particular, if $\Ac=\auta{Q,i,E,t}$ is an augmented
$\K$-automaton, an equivalence~$\Pc$ on~$Q$ is a \emph{congruence
on~$\Ac$} if and only if
\begin{gather}
\{i\}\in\Pc\EqVrgInt\e
\{t\}\in\Pc\EqVrgInt\e
\e\text{and}
\label{q.aug-con-1}
\\
\forall p,q\quantsmsp
p\Pc q \e\Longrightarrow\e
\forall a\in \Am\quantvrg
\forall C\in\Pc\quantsmsp
\sum_{r\in C}E(p,a,r)=\sum_{r\in C}E(q,a,r)
\eqpnt
\e
\label{q.aug-con-2}
\end{gather}

The virtue of adding the marker~$\dol$ and considering augmented
rather than arbitrary automata is that the sole
equation~\equnm{aug-con-2} expresses conditions described by
both equations~\equnm{con-1} and~\equnm{con-2}.
And it is the way that the property of being a congruence will be
tested by the algorithms described in the sequel.
\end{remark}

\subsection{Minimal quotients}

With \defin{out-mor}, we have defined the quotients of a
$\K$-automaton~$\Ac$.
The following proposition states that every $\K$-automaton has a
\emph{minimal quotient}.

\begin{theorem}
\label{t.min-quo}%
Among all quotients of a $\K$-automaton~$\Ac$, there exists one,
unique up to isomorphism, which has a minimal number of states and
which is a quotient of every quotient of~$\Ac$.
\end{theorem}

Once again, it is convenient to express this property in terms of
congruences, which focuses on the automaton itself rather than on its
images.
It starts with the definition of an order on equivalences, or
partitions, on a set~$Q$.

An equivalence~$\Rc$ is \emph{coarser} than an equivalence~$\Pc$ if,
for every~$C$ in~$\Pc$, there exists~$D$ in~$\Rc$ such
that~$C$ is a subset of~$D$.
It follows that every class of~$\Rc$ is the union of classes
of~$\Pc$.
The equivalences on a set~$Q$, ordered by this inclusion relation,
form a lattice, with the identity --- where every class is a
singleton --- at the top and the universal relation --- with only one
class that contains all elements of~$Q$ --- at the bottom.
The following is then a statement equivalent to \theor{min-quo}.

\begin{theorem}
\label{t.coa-con}%
Every $\K$-automaton has a unique coarsest congruence.
\end{theorem}

Indeed, the equivalence map of the morphism onto the minimal quotient
is the coarsest congruence and the coarsest congruence is the
equivalence map of the morphism onto the minimal quotient.

\begin{remark}
\label{r.min-quo-1}

It is important to stress once again that the minimal quotient
of~$\Ac$ is associated with~$\Ac$, \emph{and not} with the series
realised by~$\Ac$.
\emph{It is not} necessarily \emph{the smallest} (in terms of number
of states) $\K$-automaton that realises the series, \emph{it is not} a
canonical automaton associated with it.

For instance, the minimal quotient of a Boolean automaton is not
necessary the smallest automaton for the accepted language.
On the other hand, the minimal quotient of a deterministic Boolean
automaton is the \emph{minimal (deterministic) automaton} of the
accepted language (which may well be larger than another Boolean
automaton accepting the same language) that is, the notion we have
defined for all (possibly weighted) automata coincides with the
classical one in the case of deterministic Boolean automata.
\end{remark}

\begin{proof}[Proof of \theor{min-quo}]

Let $\msp\Ac=\aut{I,E,T}\msp$ be a $\K$-automaton of dimension~$Q$.
The \emph{proof of the existence} of a coarsest congruence on~$\Ac$
goes downward, so to speak, that is, we start from the identity on~$Q$
which is a congruence.
Let
$\msp\varphi\colon Q\rightarrow R\msp$ and
$\msp\psi\colon Q\rightarrow P\msp$
be two morphisms.
In order to prove the existence of a minimal quotient (or of a
coarsest congruence) it suffices to verify that the lower bound of the
map equivalences of~$\phi$ and~$\psi$ is a congruence.
Let
$\msp\phi'\colon R\rightarrow S\msp$ and
$\msp\psi'\colon P\rightarrow S\msp$
such that
$\msp \omega = \phi\vee\psi = \phi'\compos\phi = \psi'\compos\psi \msp$.
Hence
\begin{equation}
E\mmul\amome = E\mmul\amphi\mmul\amphip = E\mmul\ampsi\mmul\ampsip
\eqpnt
\label{q.coa-con}
\end{equation}

By definition, two states~$p$ and~$r$ of~$Q$ are equivalent
modulo~$\omega$, which we write $\msp p\,\omega\, r\msp$, if and only
if there exists a sequence
$\msp p=q_{0},q_{1}, \ldots, q_{2n}=r\msp$
of states of~$Q$ such that
$\msp q_{2i}\,\phi\, q_{2i+1}\msp$ and
$\msp q_{2i+}\,\psi\, q_{2i+2}\msp$ for
$\msp 0\leqslant i \leqslant n-1\msp$.
Rows of indices~$q_{2i}$ and~$q_{2i+1}$ of~$E\mmul\amphi$, and then
of~$E\mmul\amphi\mmul\amphip$, are equal,
rows of indices~$q_{2i+1}$ and~$q_{2i+2}$ of~$E\mmul\ampsi$, and then
of~$E\mmul\ampsi\mmul\ampsip$, are equal,
hence rows of indices~$p$ and~$r$ of~$E\mmul\amome$ are equal.
For the same reason, $\msp T_{p}=T_{r}\msp$ and~$\omega$ is a
congruence.
\end{proof}

\begin{remark}
\label{r.bis-quo}%
In the case of Boolean automata also, the characterisation of
bisimulation follows from \remar{mor-bis}.
\begin{proposition}
\label{p.bis-quo}%
Two $\K$-automata are bisimilar if and only if their minimal quotients
are isomorphic.
\EoP
\end{proposition}
\end{remark}

The minimal quotient not only exists but is also effectively, and
efficiently computable,
as we see in the next sections.

\section{The computation of the minimal quotient}
\label{s.com-tru}%

To describe the algorithms computing the minimal quotient of a
weighted automaton, it is convenient to consider the augmented
automaton.
From now on, every automaton is therefore augmented, the alphabet
is~$\Adol$, and the states which are different from the initial or
final states are called \emph{true states}.

\subsection{Refinement algorithms and signatures}

The principle of algorithms which are studied in this paper is
partition refinement.
The algorithms start with the coarsest equivalence (where all true
states are gathered in one class), and split classes which are
inconsistent with the constraints of a congruence.
In order to split classes, we use a criterion on states of~$\Ac$, which
we call \emph{signature}, and which, given a partition, tells if two
states in a same class should be separated in a congruence of~$\Ac$.
This notion of signatures
has been used in~\cite{BealCroc08} for the minimisation of
\emph{incomplete} deterministic Boolean automata, the `first' example for which the classical
minimisation algorithms for complete deterministic Boolean automata have to be adapted.

The signature is first defined with respect to a class.
\begin{definition}
The \emph{signature} of a state~$p$ of a
$\K$-automaton~$\Ac=\aut{Q,i,E,t}$ with respect to a subset~$D$ of~$Q$
is \emph{the map}~$\sigf{p,D}$ from~$\Adol$ to~$\K$, defined by:
\begin{equation}
\sig{p,D}{a} = \sum_{q\in D}E(p,a,q)
\eqpnt
\notag
\end{equation}
\end{definition}

For every state~$p$ and every subset~$D$, the domain of~$\sigf{p,D}$
is the set of labels of transitions from~$p$ to some state of~$D$.
In particular, if a letter~$a$ does not belong to the domain
of~$\sigf{p,D}$, then~$\sig{p,D}{a}=0$.
When we want to explicitly describe~$\sigf{p,D}$, we write it as a set
of elements of the form $a\mapsto k$, where~$a$ is in the domain
of~$\sigf{p,D}$ and~$\sig{p,D}{a}=k$.

\medskip
Two states~$p$ and~$q$ cannot be equivalent with respect to some
congruence~$\Pc$ of~$Q$ if their signatures differ
for some class~$D$ of~$Q$.
In order to compare states, it can be useful to consider the
\emph{global} signature which is the aggregation of signatures with
respect to all classes of the partition.
That is, for a given partition~$\Pc$ and for every state~$p$,
\begin{equation}
\glsigg{p}=
\Defi{(D,a)\mapsto k}{ D\text{ class of }\Pc, a\mapsto k\in\sigf{p,D}}
\eqpnt
\notag
\end{equation}

\subsection{The Protoalgorithm}
The computation of the coarsest congruence of
an automaton~$\Ac=\aut{Q,i,E,t}$ \emph{goes upward}.
It starts with
$\msp\Pc_{0}=\big\{\{i\},Q,\{t\}\big\}\msp$,
the coarsest possible equivalence.
Every step of the algorithm splits some classes of the
current partition, yielding an equivalence which is \emph{thiner},
\emph{higher} in the lattice of equivalences
of~$Q\cup\{i,t\}$.

\medskip
It follows from Proposition~\ref{p.con} that a partition~$\Pc$ is a congruence
if and only if
\begin{equation}
\forall C\in\Pc\quantvrg \forall p,q\in C\quantvrg
\glsigg{p} = \glsigg{q}
\eqpnt
\eee\eee
\notag
\end{equation}
Thus~$\Pc$ is a congruence if and only if, for every pair~$(C,D)$ of
classes of~$\Pc$, all states~$p$ in~$C$ have the same signature
with respect to~$D$.
A pair~$(C,D)$ for which this property is not satisfied is called a
\emph{splitting pair}.
The equivalence on~$C$ induced by the signature with respect to~$D$,
called \emph{the split of~$C$ by~$D$} and denoted by~$\spltf{C,D}$,
can be computed:
\begin{equation}
\forall p,q\in C\quantsp
\splt{C,D}{p} = \splt{C,D}{q}
\e\Longleftrightarrow\e
\sigf{p,D} = \sigf{q,D}
\eqpnt
\ee
\notag
\end{equation}
The split of class~$C$ by class~$D$ of~$\Pc$ leads to a new
equivalence on~$Q$: $\Pc\splteq\spltf{C,D}$,
and the protoalgorithm runs as follows:

\begin{tabbing}
 xxxxxx \= xxx \= xxx \= xxx \= xxx \= xxx \=\kill
\> $\Pc:=\Pc_0$\\
\> \keyw{while} there exists a splitting pair~$(C,D)$ in~$\Pc$\\
\> \> $\Pc:=\Pc\splteq\spltf{C,D}$
\end{tabbing}

When there are no more splitting pairs, the current equivalence is a
congruence.

\begin{proposition}
\label{p.coa-con}
At every iteration of the protoalgorithm, the equivalence~$\Pc$ is
coarser than or equal to the coarsest congruence.
\end{proposition}

\begin{proof}
In the initial partition~$\Pc_0$, all true states are in the same
class, thus~$\Pc_0$ is coarser than or equal to any congruence.

\medskip
Let~$\Pc$ be a partition computed at one iteration; we assume that
$\Pc$ is coarser than or equal to the coarsest congruence~$\Cc$, and we
consider~$\Pc'$ computed at the next iteration.
There exists a splitting pair~$(C,D)$, such that
$\Pc'=\Pc\splteq\spltf{C,D}$.
If $\Pc'$ is not coarser than or equal to~$\Cc$, then there exists~$C_1$ and $C_2$
in $\spltf{C,D}$, $p_1$ in~$C_1$ and $p_2$ in~$C_2$ which belongs to
the same class in~$\Cc$.
Since~$\Pc$ is coarser than or equal to~$\Cc$, $D$ is the union of some
classes~$(D_i)$ in~$\Cc$; $p_1$ and $p_2$ are $\Cc$-equivalent, hence,
for every~$i$,
\begin{equation}
\sigf{p_1,D_i}=\sigf{p_2,D_i}
\eqpnt
\end{equation}
Therefore, $\sigf{p_1,D}=\sigf{p_2,D}$, which is in contradiction with
the fact that~$p_1$ and $p_2$ are in different classes after the split
of $C$ with respect to~$D$.
Thus $\Pc'$ is coarser than or equal to~$\Cc$ and, by induction, the proposition
holds.
\end{proof}

\begin{corollary}
\label{c.coa-con}
At the end of the protoalgorithm, the equivalence~$\Pc$ is
equal to the coarsest congruence. 
\EoP
\end{corollary}

The procedure described by the protoalgorithm is not a true algorithm,
in the sense that, in particular, it does not tell how to find a
splitting pair, nor how to implement
the function~$\spltop$ to make it efficient.
The main difference between the two algorithms described
in the sequel is the selection of the splitting pairs.
\begin{itemize}
\item the first algorithm
iterates over classes~$C$ and considers in the same iteration all the
pairs~$(C,D)$ where~$D$ is any class which contains some
\emph{successors} of states of~$C$.
This leads to split $\Cc$ in classes which are consistent with every
class~$D$.
We call it the \emph{Domain Split Algorithm}
(\DSA for short).

\item the second algorithm
iterates over classes $D$ and considers in the same iteration all the
pairs $(C,D)$ where~$C$ is any class which contains some
\emph{predecessors} of states of~$D$.
This leads to split several classes in such a way that all classes
become consistent with the class~$D$.
We call it the Predecessor Class Split
Algorithm (\PCSA for short).
\end{itemize}

\subsection{The Domain Split Algorithm}
\label{ss.domain}%

The Domain Split Algorithm is an extension of the classical Moore
algorithm~\cite{Moor56} for the minimisation of deterministic Boolean
automata.
At each iteration of this algorithm, a class~$C$ is processed.
The \emph{global signature} of every state of~$C$ with respect to the
current partition is computed, and~$C$ is split accordingly.
Compared to the protoalgorithm, the Domain Split Algorithm amounts
therefore to consider at  
the same time all the pairs~$(C,D)$,
where~$C$ is the current class, and~$D$ is any class of the current
partition (actually, only classes which contain some successors of
states of~$C$ are considered).
Notice that it is mandatory to recompute signatures at each iteration,
since they depend on the partition which is changing.

At the beginning of each iteration, the current class~$C$ is extracted
from a queue.
At the beginning of the algorithm, the partition~$\Pc_{0}$ is
$\{\{i\},Q,\{t\}\}$, and~$Q$ is inserted in the queue.
At the end of each iteration, the classes obtained from the split
of~$C$ --- or~$C$ itself if it has not been split --- are inserted into
the queue, except the singletons that cannot be split.

To insure termination, iterations are gathered to
form rounds.
A round ends when all classes that were inside the queue at the
beginning of the round have been extracted.
Hence the number of iterations in a round is equal to the number of
classes which are in the queue at the beginning of the round.
Notice that, for every state~$p$, there is at most one class
containing~$p$ that is processed during a round.
If there is no split during a round, the algorithm has checked that the
partition is a congruence and halts.

Otherwise, the partition is strictly refined and the number of classes
strictly increases.
At the beginning of the algorithm, there are three classes, and the
maximal number of classes is~$n+2$, where $n$ is the number
of true states of the automaton.
Hence, including the last round, there are at most~$n$ rounds, and,
for every state, at most~$n$ global signatures are computed.

\begin{example}[continued]
On the automaton~$\Ac_1$, the partition~$\Pc_{0}$
is initialised with $D_1=\{i\}$, $D_2=\{p,q,r\}$, $D_3=\{t\}$.
Class~$D_2$ is put into the queue (the other classes are singletons and
cannot be split).
For every state of~$D_2$ and for every class~$D$, one can compute the
signature of this state with respect to~$D$.
For instance, the signature of~$p$ with respect to~$D_2$ is
\begin{align*}
\sig{p,D_2}{a}&= E(p,a,p)+E(p,a,q)+E(p,a,r)=-1+0+0=-1\eqvrg\\
\sig{p,D_2}{b}&= E(p,b,p)+E(p,b,q)+E(p,b,r)=0+-1+2=1\eqpnt
\end{align*}
From the signature with respect to each class, the global signature of
the state can be computed.

The global signature (with respect to~$\Pc_{0}$) of states in~$D_2$ is:
\begin{align*}
\glsigg{p}& =\{(D_2,a)\mapsto -1,(D_2,b)\mapsto 1\}\eqvrg\\
\glsigg{q}& =\{(D_2,a)\mapsto 2,(D_2,b)\mapsto 1,(D_3,\mmark)\mapsto 1\}\eqvrg\\
\glsigg{r}& =\{(D_2,a)\mapsto 2,(D_2,b)\mapsto 1,(D_3,\mmark)\mapsto 1\}\eqpnt
\end{align*}
States~$q$ and $r$ share the same global signature which is different
from the global signature of~$p$.
The class~$D_2$ is then split into two classes~$D_{21}=\{p\}$ and
$D_{22}=\{q,r\}$ and the new partition is
$\Pc_{1}=\{\{i\},\{t\},\{p\},\{q,r\}\}$.
The new round starts and the class~$D_{22}$ is put in the queue.
The global signature of states (with respect to~$\Pc_{1}$) in~$D_{22}$ is:
\begin{align*}
\glsigg{q}& =\{(D_{21},a)\mapsto 1,(D_{22},a)\mapsto 1,(D_{22},b)\mapsto 1,(D_3,\mmark)\mapsto 1\}\\
\glsigg{r}& =\{(D_{21},a)\mapsto 1,(D_{22},a)\mapsto 1,(D_{22},b)\mapsto 1,(D_3,\mmark)\mapsto 1\}
\end{align*}
Both states have the same global signature, thus the class is not split.
The round ends without any splitting, hence the current
partition~$\Pc_{1}$ is a congruence.
\end{example}

\subsection{Predecessor Class Split Algorithm}
\label{ss.pred-alg}%

The Predecessor Class Split Algorithm is \emph{inspired} by the
Hopcroft algorithm for the minimisation of deterministic Boolean
automata.
At each iteration of this algorithm, a class~$D$ is processed.
For every state~$p$ which is a predecessor of some states of~$D$, the
signature of~$p$ with respect to~$D$ is computed, and~$C$ is split accordingly.
Compared to the protoalgorithm, \PCSA amounts therefore to consider
at the same time all the pairs~$(C,D)$, where $D$ is the current
class, and $C$ is any class of the current partition (actually, only
classes which contain some predecessors of states of~$D$ are
considered).

Like the \DSA, the \PCSA uses a queue to schedule the process of classes.
At the beginning of the algorithm, every class of~$\msp\Pc_0\msp$ is put
into the queue, except~$\{i\}$ since there is no predecessor of~$i$.
If a class~$C$ is split, every new class is put in the queue; if~$C$
was in the queue, it is replaced by its subclasses.
The algorithm halts when the queue is empty.
Actually, when the partition is a congruence, every class extracted
from the queue does not induce any splitting, hence there is no more
insertion of classes.

Apart from the initialisation of the queue, every
class which is inserted in the queue comes from a split.
Hence, for every state~$p$, the number of times that a class
containing~$p$ is extracted from the queue is at most equal to~$n$, where $n$
is the number of true states of the automaton.
Notice that there is no notion of rounds in the \PCSA.

\begin{example}[continued]

On the automaton~$\Ac_1$, the equivalence is initialised
with~$D_1=\{i\}$, $D_2=\{p,q,r\}$, $D_3=\{t\}$.
Classes~$D_2$ and~$D_3$ are put in the queue ($i$ has no predecessor,
thus~$D_1$ can not split any class).
Assume that~$D_2$ is considered first.
The signatures of predecessors of states in~$D_2$ are:
\begin{align*}
\sigf{i,D_2}=&\{\mmark\mapsto 3\}\eqvrg&
\sigf{p,D_2}=&\{a\mapsto -1, b\mapsto 1\}\eqvrg\\
\sigf{q,D_2}=&\{a\mapsto 2, b\mapsto 1\}\eqvrg&
\sigf{r,D_2}=&\{a\mapsto 2, b\mapsto 1\}\eqpnt
\end{align*}
All states in the class~$D_1=\{i\}$ have the same signature ($D_1$ is
a singleton!).
In the class~$D_2$, every state is met, but~$p$ has a signature which
is different from the signature common to~$q$ and~$r$.
Hence~ $D_2$ is split into two classes~$D_{21}=\{p\}$
and~$D_{22}=\{q,r\}$ which are inserted in the queue.

\medskip
The next class which is extracted from the queue is~$D_3$, and the
signatures of predecessors of states in~$D_3$ are:
\begin{align*}
\sigf{q,D_3}=&\{\mmark\mapsto 1\}\eqvrg&
\sigf{r,D_3}=&\{\mmark\mapsto 1\}\eqpnt
\end{align*}
Both states have the same signature and there is no other state in
$D_{22}$, thus there is no split.

\eject

The class processed in the next iteration is~$D_{21}$, and the
signatures of predecessors of states in~$D_{21}$ are:\vspace*{-2mm}
\begin{align*}
\sigf{i,D_{21}}=&\{\mmark\mapsto 2\}\eqvrg&
\sigf{p,D_{21}}=&\{a\mapsto -1\}\eqvrg\\
\sigf{q,D_{21}}=&\{a\mapsto 1\}\eqvrg&
\sigf{r,D_{21}}=&\{a\mapsto 1\}\eqpnt
\end{align*}
States~$i$ and~$p$ are already in classes which are singletons;
states~$q$ and~$r$ have the same signature and there is no other state
in~$D_{22}$, thus there is no split.

\medskip
The next processed class is~$D_{22}$, and the signatures of
predecessors of states in~$D_{22}$ are:
\begin{align*}
\sigf{i,D_{22}}=&\{\mmark\mapsto 1\}\eqvrg&
\sigf{p,D_{22}}=&\{b\mapsto 1\}\eqvrg\\
\sigf{q,D_{22}}=&\{a\mapsto 1, b\mapsto 1\}\eqvrg&
\sigf{r,D_{22}}=&\{a\mapsto 1, b\mapsto 1\}\eqpnt
\end{align*}
States~$i$ and~$p$ are already in classes which are singletons;
states~$q$ and~$r$ have the same signature and there is no other state
in $D_{22}$, thus there is no split.

\medskip
The queue is empty; the algorithm halts and the current
partition is a congruence.
\end{example}

\section{Complexity of the refinement algorithms}
\label{s.imp-alg}%

The high-level description of the \DSA and \PCSA
have shown that every state is processed at most~$n$ times, where~$n$
is the number of true states.
Processing a state consists in running over its
outgoing transitions (for \DSA) 
or its incoming
transitions (for \PCSA) 
to compute signatures.
Globally, every transition of the automaton is thus considered at
most~$n$ times.

To be efficient, the time to compute the signatures must be linear in
the \emph{number of transitions} involved in the computation.
The key point in the algorithms is the ability to compute signatures
in such a way that the signatures of two states can be compared in
linear time.
Usually, to compare two lists efficiently, a preliminary step consists
in sorting them.
To reach the linear complexity, a true sort is not affordable.
Hence, we present here a \emph{weak sort}~\cite{Paig94} which allows to
compare signatures in linear time.

\subsection{The weak sort}

Let~$f$ be an evaluation function from a set~$X$ to a set~$Y$.
In our framework, the evaluation function is the signature.
If~$Y$ is totally ordered, sorting a list of elements of~$X$ with
respect to~$f$ consists in computing a permutation of the list such
that the elements of the list are ordered with respect to the
evaluation~$f$.
If such a list is sorted, the elements with the same evaluation are
contiguous and it is then easy to gather them.
Following the ideas given in~\cite{Paig94}, it is not necessary to
fully sort the list, and we say that a \emph{list} of elements of~$X$
is \emph{weakly sorted} (with respect to~$f$) if the elements with the
same evaluation by~$f$ are \emph{contiguous}.
It is then easy to compute the map equivalence of~$f$.

\medskip
Both \DSA and \PCSA 
are based on a weak sort with respect to signatures.
Since signatures are themselves lists, nested weak sorts are necessary
to gather states with the same signatures such that equal signatures
are represented by the same lists.

In a first step, and for every state~$p$, transitions outgoing
from~$p$ with the same label and with destinations in the same class
must be gathered.

Moreover, for two states with the same signature, the list of pairs
(label $\mapsto$ weight) that form the signature must appear in the same
ordering, in such a way the equality test is insured to be efficient.

In a second step, the weak sort is used to gather states
with the same signature and to form the new classes.

\medskip
A \emph{bucket sort algorithm}~\cite{CorLeiRivSte09} realises a weak sort with linear
complexity in the size of~$X$.
Assume first that arrays indexed by~$Y$ can be managed.
Let~$T$ be such an array where elements are lists of values in~$X$.
The algorithm iterates over~$X$ and stores every~$x$ in~$X$ in the
list~$T[f(x)]$.
Finally, by iterating over~$T$, all the lists are concatenated in
order to compute a weakly sorted list.

This naive description hides two issues.
First, $Y$ may be large, and the initialisation of every element
of~$T$ to the empty list is linear in~$Y$, which can be much larger
than~$X$.
The second issue is related to the first one, there may be a few~$y$
in~$Y$ which are images of some~$x$ by~$f$, and it is too expensive to
iterate over all elements of~$Y$ to concatenate the lists.

\medskip
The \emph{hash maps} are a solution to the first issue
(see~\cite{Knuth98, CorLeiRivSte09} for instance).
They allow to avoid the blowing up of memory in the case where~$Y$ is
huge.
As the amortised access time is constant, it allows to implement the
weak sort in linear time.
To solve the second issue, the keys of the hash map can be linked in
order to efficiently iterate over the elements of~$f(X)$ (the data
structure is then equivalent to the \emph{linked hash maps}
implemented for instance in Java~\cite{JavaLHM21}).

\subsection{Application to the Domain Split Algorithm}

The computation of the global signature $\GSig{p}$ of the states of
the current class~$C$ requires two steps.

\medskip
First, an array indexed by~$\Adol\times\Pc$ stores lists of pairs
in~$Q\times\K$.
For every state~$p$ in~$C$, and for every transition
$p\overset{a|k}{\longrightarrow}q$, the pair~$(p,k)$ is
inserted in a list $\meet{a,D}$, where~$D$ is the class of~$q$.
This insertion is special in the case where $\meet{a,D}$ is not empty
and its last element is a pair $(p',k')$ with $p'=p$:
$k'$ is then updated to~$k'+k$; if~$k'+k$ is zero, the pair is
removed from the list.
The second step builds the global signature itself.
For every useful index~$(a,D)$, for
every~$(p,k)$ in $\meet{a,D}$, inserts $(a,D)\mapsto k$ into $\GSig{p}$.

We see that the elements in $\GSig{p}$ follow an ordering which is
consistent with the ordering of iteration on $\meetop$.
Two states with the same global signatures have therefore the same
list.

\subsection{Application to the Predecessor Class Split Algorithm}

For \PCSA, the computation of the signature is slightly
easier.
The first array is only indexed by~$\Adol$.
For each~$q$ in~$D$, for every transition
$p\overset{a|k}{\longrightarrow}q$,
$(p,k)$ is inserted in a list $\meet{a}$.
Then, for every key~$a$ of $\meetop$, for every~$(p,k)$ in $\meet{a}$,
$a\mapsto k$ is inserted in $\sigf{p,D}$ with the same special insertion
used in \DSA. 
Notice that this special insertion occurs here in the second step
whereas it is used in the first step in \DSA.

\subsection{Splitting of classes}

In \DSA 
the current class is split, the signature of every state
of the class is considered, and it is not difficult to split the class
in linear time with respect to its size.

In \PCSA 
the current class~$D$ is the splitter; it may induce the splitting of
several classes, and all the operations must be performed in a time
which is linear with respect to the number of transitions incoming to
states of~$D$.
In particular, the splitting of a class must be performed in a time
which is linear with the number of the states which are predecessors
of a state of~$D$ --- that may be much smaller than the size of the
class.

To this end, it is required that the deletion of any element of a class
can be performed in constant time.
Thus classes are implemented by double linked lists, and an array
indexed by states gives, for each state, a pointer to the location of
the state in its class.

\subsection{Analysis of both algorithms}

In \DSA, as seen in Section~4.3,
every state may be considered~$n$ times, where~$n$ is the number of
states.
The computation of the global signature of a state requires a number
of operations which is linear with the number of transitions outgoing
from this state.
Every iteration requires a time which is linear with the number of
transitions outgoing from the current class.
Finally, the time complexity of \DSA is
in~$\grando{n(m+n)}$, where~$n$ is the number of states and~$m$ is the
number of transitions of the automaton, provided that each operation
and the computation of a hash for the weights is in constant time.

In \PCSA, every state may be considered~$n$ times.
The computation of the signatures during
an iteration is linear with the number of transitions incoming to the
current class.
Finally, the time complexity of \PCSA is in~$\grando{n(m+n)}$,
under the same conditions as above.

\begin{theorem}
The Domain Split Algorithm and the Predecessor Class Split Algorithm
compute the minimal quotient of a
$\K$-automaton with~$\xmd n\xmd$ states and~$\xmd m\xmd$ transitions
in time $\grando{n(m+n)}$.
\end{theorem}

In the case of a deterministic Boolean automaton,
where~$\msp m=\alpha\xmd n\msp$, with~$\alpha=|{A}|$, we get back to the
classical~$\grando{\alpha\xmd n^{2}}$ complexity of the Moore
minimisation algorithm (\cf~\cite{HopcMotwUllm06,BersEtAl21} for instance).
It may seem strange that \PCSA which has been said to be inspired by
Hopcroft's algorithm has the same complexity.
This is explained in the next section.

\section{Conditions for an $\grando{m\log n}$ algorithm}
\label{s.fas-PCS}%

Hopcroft's algorithm can be seen as an improvement of \PCSA for
complete Boolean deterministic automata.
Its time complexity is~$\grando{\alpha\xmd  n\log n}$ (\cf for
instance~\cite{Grie73,BersCart04}), where~$\alpha$ is the size of the
alphabet;
this algorithm has been
extended to incomplete DFA~\cite{BealCroc08,ValmLeht08} with
complexity~$\grando{m\log n}$.

The strategy: ``All but the largest", introduced in~\cite{Hopc71},
can be applied to improve \PCSA in some cases that
we now study.

At every step of \PCSA, some classes~$C$ are evaluated
(through signatures) with respect to the current splitter~$D$.
If the class~$D$ is split into several classes, $D_1,\ldots, D_k$,
all these classes are processed as splitters in further iterations.

The idea of the ``All but the largest" strategy is that it is useless
to process the last of the subclasses of~$D$ because after the splits
induced by~$D$ itself and the splits induced by all the other
subclasses, this subclass does not induce any new split.
If this is true, one can choose which subclass is not processed; in order
to get a better complexity, the strategy commands to choose the larger
one.

\subsection{Simplifiable signatures}
A sufficient condition to apply this strategy is that the signatures
with respect to the last subclass can be deduced from the signatures
with respect to~$D$ and the signatures with respect to the other
subclasses.

\medskip
The signatures are equipped with the pointwise addition:
for every~$a$ in~$\Am$,
\begin{equation}
    \forall a\in\Am\quantsp
\big(\sigf{p,D}+\sigf{p,D'}\big)(a)=\sig{p,D}{a}+\sig{p,D'}{a}
\eqvrg
\notag
\end{equation}
and if~$D$ is a subset of~$Q$ and~$\psi$ a partition of~$D$, then it
holds:
\begin{equation}
\sigf{p,D}=\sum_{D'\text{ class of }\psi}\sigf{p,D'}
\eqpnt
\notag
\end{equation}

\begin{definition}
An automaton has \emph{simplifiable signatures} if, for every
subset~$D$ of~$Q$ and every subset~$C$ of~$D$, and for every pair of
states~$p,q$, it holds
\begin{equation}
\sigf{p,D}=\sigf{q,D}\text{ and }
\sigf{p,C}=\sigf{q,C}\Longrightarrow
\sigf{p,D\setminus C}=\sigf{q,D\setminus C}.
\notag
\end{equation}
\end{definition}

A commutative monoid~$(M,\oplus)$ is \emph{cancellative} if
for every $a$, $b$, and~$c$ in~$M$,
$a\oplus b=a\oplus c$ implies~$b=c$. In particular,
every group is cancellative, and if~$\K$ is a ring, the additive monoid~$(\K,+)$ is cancellative.
\begin{lemma}
Let~$\K$ be a semiring such that the additive monoid~$(\K,+)$ is cancellative.
Then every $\K$-automaton has simplifiable signatures.
\end{lemma}

For other weight semirings, the simplifiability of signatures depends
on the automaton.
If~$\Ac$ is a deterministic\footnote{%
called \emph{sequential} in~\cite{Saka09}.}
 $\K$-automaton, that is, if
for every state~$p$ and every letter~$a$, there is at most one
transition outgoing from~$p$ with label~$a$, then
the signatures are simplifiable, independently of~$\K$, since it
holds:
\begin{equation}
\forall p\in Q\quantvrg\forall a\in\Am\quantsmsp
\sig{p,D\setminus C}{a}=\begin{cases}
\sig{p,D}{a}&\text{ if } \sig{p,C}{a}=0_\K\eqvrg\\
0_\K &\text{ otherwise}\eqpnt
\end{cases}
\e
\notag
\end{equation}

Typically, incomplete deterministic Boolean automata, as considered
in~\cite{BealCroc08}, have simplifiable signatures whereas general
Boolean automata have not.

\begin{example}
Let $\Ac_2$ be the nondeterministic Boolean automaton of Figure~\ref{fig.nfa}.
It holds:
\begin{align*}
\sig{p,\{r,s\}}{a}&=1 \qquad& \sig{q,\{r,s\}}{a}&=1\\
\sig{p,\{s\}}{a}&=1 \qquad& \sig{q,\{s\}}{a}&=1\\
\sig{p,\{r\}}{a}&=1 \qquad& \sig{q,\{r\}}{a}&=0.
\end{align*}
The signature with respect to~$\{r\}$ cannot be deduced from the
signature with respect to~$\{r,s\}$ and~$\{s\}$.
Hence, $\Ac_2$ has not simplifiable signatures.
\begin{figure}[h]
\vspace*{-1mm}
\centering
\includegraphics[width=4.15cm]{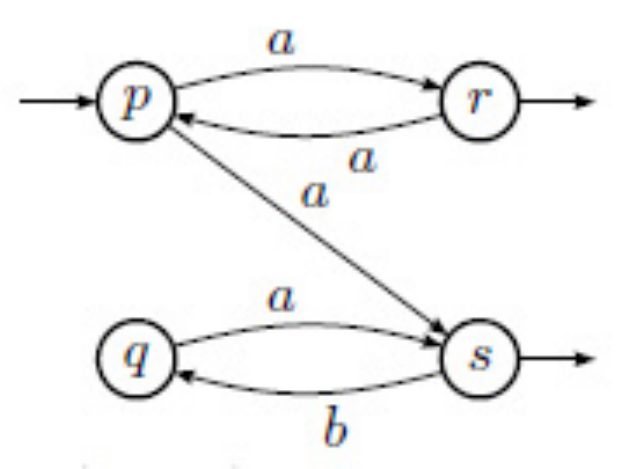}\vspace*{-4mm}
\caption{The nondeterministic automaton $\Ac_2$}\label{fig.nfa}
\end{figure}\vspace*{-3mm}
\end{example}

\subsection{The ``All but the largest'' strategy}
If~$\Ac$ is an automaton with simplifiable signatures,
\PCSA can be improved. We call this improved algorithm
Fast Predecessor Class Split algorithm (\FPCSA).

For every
class~$C$ which is split with respect to a class~$D$, if~$C$ is
not already in the queue, all the classes in~$\spltf{C,D}$ except one of the largest are put in the queue. Notice that finding the largest class can be done in linear time with respect to the size of~$\spltf{C,D}$: the size of every subclass containing some predecessor of a state of~$D$ is computed at the same time as the signatures, the size of the class containing the other states is the difference of the size of~$C$ with the sum of the sizes of the other subclasses.

Actually, since~$C$ is not in the queue, the splitting of
classes with respect to~$C$ has already been considered: for every~$D$
in the current partition, $\sigf{p,C}$ is the same for all~$p$
in~$D$.
Let~$C_1$ be some subclass of~$C$; if~$(D,C_1)$ is a splitting pair
for some class~$D$, then, since signatures are simplifiable, there
exists some other class~$C_2$ in~$\psi$ such that~$(D,C_2)$ is also a
splitting pair.

\medskip

Let~$c(k)$ be the maximal number of times that a state~$p$ which
belongs to a class~$D$ of size~$k$ that is removed from the queue will
appear again in classes removed from the queue.
A class~$D'$ containing~$p$ will be inserted only if~$D'$ results from
a split of~$D$ and there exists another class~$D''$ that results from
the same split and whose size is at least as large as the size of~$D$.
Hence, the size of~$D'$ is at most $k/2$.
Finally $c(k)\leqslant 1+c(k/2)$, and, since a singleton class will
not be split, $c(1)=0$; therefore~$c(k)$ is in $\grando{\log k}$.
Thus, the complexity of the algorithm is in $\grando{(m+n)\log n}$.
This complexity meets the complexity of the Hopcroft
algorithm~\cite{Hopc71} for the minimisation of complete deterministic
automata which is in~$\grando{\alpha n\log n}$, where~$\alpha$ is the
size of the alphabet; in this case, $m=\alpha n$.

\begin{theorem}
If~$\Ac$ is a $\K$-automaton with simplifiable signatures, \FPCSA
computes the minimal quotient of~$\Ac$ in time~$\grando{(m+n)\log n}$.
\end{theorem}

In the case of nondeterministic Boolean automata, as we have seen, the
signatures are not simplifiable and the improvement from \PCSA to
\FPCSA is therefore not warranted.
Nevertheless, the Relation Coarsest Refinement algorithm described
in~\cite{PaigTarj87} can be extended in order to compute the minimal
quotient in time~$\grando{m\log n}$, as explained in~\cite{Fern89}.
For weighted nondeterministic automata over other semirings with no
cancellative addition, like the $(\min,+)$-semiring, it is an open
to know whether there exists an algorithm in time~$\grando{m\log n}$ for the computation of the
minimal quotient.

\section{Examples and benchmarks}
\label{s.ben-mar}%

\DSA, \PCSA and \FPCSA are implemented in the \textsc{Awali}
library~\cite{Awali}.  We present here a few benchmarks to compare
their respective performances and to check that their execution time
is consistent with their asserted complexity.
Benchmarks have been run on an iMac Intel Core i5 3,4GHz, compiled
with Clang~9.0.0.

\medskip
First, we study a family of automata which is an adaptation of a
family used in~\cite{CasResSci10} to show that the Hopcroft algorithm
requires $\gtheta{n\log n}$ operations.

Let $\varphi$ be the morphism defined on $\{a,b\}^*$ by
$\varphi(a)=a\xmd b$ and $\varphi(b)=a$;
for instance
$\varphi(a\xmd b\xmd a\xmd a\xmd b)=
a\xmd b\xmd a\xmd a\xmd b\xmd a\xmd b\xmd a$.
The $k$-th Fibonacci word is $w_k=\varphi^k(a)$; its length is equal to
the $k$-th Fibonacci number~$F_k$,
hence it is in~$\gtheta{\bigl(\frac{1+\sqrt{5}}{2}\bigr)^k}$.
Notice that for every~$k\geqslant 2$, $w_k=w_{k-1}.w_{k-2}$.
Let $\mathcal{F}_k$ be the automaton with one initial state and a
simple circuit around this initial state with label~$w_k$
(all states are final).

\renewcommand{\arraystretch}{1.2}
\begin{table}[!h]
\vspace*{-5mm}
\caption{Minimisation of $\mathcal{F}_k$}
\label{tab.fib}
\centering
\scalebox{0.89}{
\begin{tabular}{|c|c|cccccc|}
\hline
& $k$       &  14  &  17  &  20 &  23  &  26  &  30 \tabularnewline
&$F_k$      &  $987$  & $4\xmd 181$ & $17\xmd 711$ & $75\xmd 025$ & $317\xmd 811$ & $2\xmd 178\xmd 309$  \tabularnewline
\hline
\DSA & $t$ (s)       & 0.42  & 7.37     & 139& \multicolumn{3}{c|}{\_}\tabularnewline
&$10^{-7} t$/$F_k^2$ &  4.3  & 4.2 & 4.4& \multicolumn{3}{c|}{}\tabularnewline
\hline
\PCSA &$t$ (s)        & 0.010  & 0.045 & 0.257 & 1.36 & 73       & 257     \tabularnewline
&$10^{-7} t$/$k\xmd F_k$ &7.2&6.3&7.3& 7.6   & 6.7       & 7.5      \tabularnewline
\hline
\FPCSA & $t$ (s)                          & 0.006   & 0.025  & 0.140   & 0.70   & 41 & 139      \tabularnewline
&$10^{-7} t$/$k\xmd F_k$ &4.2&3.5& 3.9 & 3.8  & 3.5    & 3.7      \tabularnewline
\hline
\end{tabular} }\vspace{1mm}
\end{table}

We observe on the benchmarks of Table~\ref{tab.fib} that the running
time of the \DSA is quadratic,
while the running time of both \PCSA and \FPCSA
algorithms are in~$\gtheta{k\xmd F_k}$
({\it i.e.} $\gtheta{F_k\log F_k}$, where~$F_k$ is the number of states).

The second family is an example where \PCSA and \FPCSA have not the
same complexity.
Notice that these automata are acyclic and there may exist faster
algorithms ({\it cf.}~\cite{BersEtAl21} for specific algorithms for
acyclic Boolean deterministic automata), but this is out of the scope
of this paper.

The $n$-th ``Railroad'' automaton has $2n$ states numbered from~$1$
to~$2n$, and for every $p$ in~$[1;n-1]$, there are transitions from
states~$2p-1$ and~$2p$ to states~$2p+1$ and~$2p$, as described by
Figure~\ref{fig.rail}.
The state~$1$ is initial and both~$2n-1$ and~$2n$ are final.

\begin{figure}[!h]
\centering
\includegraphics[width=5.36cm]{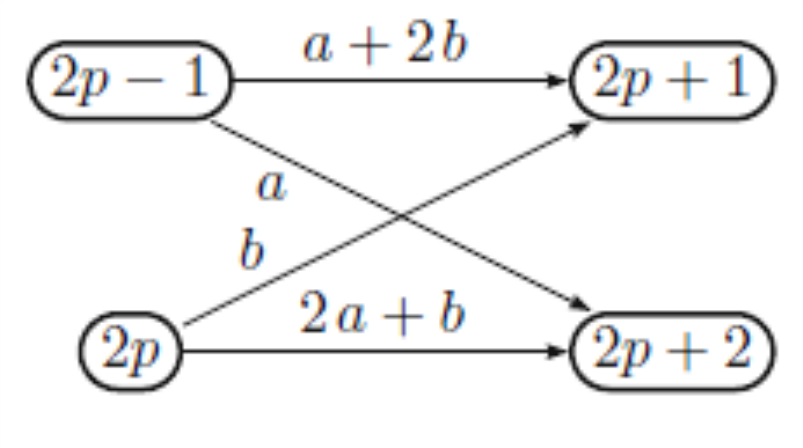}\vspace*{-3mm}
\caption{The transitions of the Railroad $\Z$-automaton}\label{fig.rail}
\end{figure}

The benchmarks of the minimisation of ``Railroad'' automata are shown
on Table~\ref{tab.rail}.
On these automata, if a temporary class contains all states
between~$1$ and~$2k$, it is split into one class~$[1;2k-2]$ and one
class~$\{2k-1,2k\}$: the size of the classes lowers slowly.
On these examples, \DSA and \PCSA are therefore quadratic.

\begin{table}[!h]
\vspace*{-4mm}
\renewcommand{\arraystretch}{1.12}
\centering
\caption{Minimisation of Railroad($n$)}\label{tab.rail}
\scalebox{0.89}{
\begin{tabular}{|c|c|cccccc|}
\hline
$\strut_{\strut}^{\strut}$
&$n$ & $2^{10}$ & $2^{12}$ & $2^{13}$ & $2^{14}$ & $2^{15}$ & $2^{22}$\tabularnewline
\hline
\DSA & $t$ (s)        & 3.29 & 53.2 & 214 & \multicolumn{3}{c|}{\_}\tabularnewline
&$10^{-6} t$/$n^2$ & 3.1 & 3.2 & 3.2 & \multicolumn{3}{c|}{}\tabularnewline\hline
\PCSA & $t$ (s)      & 0.31 & 4.92 & 20.5 & 86.1 & 346& \multicolumn{1}{c|}{\_}\tabularnewline
& $10^{-7} t$/$n^2$& 3.0  & 2.9 & 3.1 & 3.2 & 3.2 & \multicolumn{1}{c|}{}\tabularnewline\hline
\FPCSA & $t$ (s)         & $0.008$ & $0.030$ & $0.061$ & 0.12 & 0.24 & 30.8\tabularnewline
& $10^{-6} t$/$n$ & 7.8  & 7.3 & 7.4 & 7.3 & 7.3 & 7.3\tabularnewline\hline
\end{tabular} }\vspace{1mm}
\end{table}

In \FPCSA, when this splitting occur, the largest class~($[1;2k-2]$)
is not put in the queue for further splittings; therefore, except at
the first round, all splitters are pairs of states and the algorithm
is linear.

\end{document}